\documentclass[12pt,notitlepage]{article}%
\usepackage{amsthm}
\usepackage{setspace}
\usepackage{eurosym}
\usepackage{subfigure}
\usepackage{xcolor}
\usepackage{amssymb}
\usepackage{mathrsfs}
\usepackage{graphicx}
\usepackage[colorlinks,linkcolor=links,citecolor=cites,urlcolor=MyDarkBlue]%
{hyperref}
\usepackage{amsfonts}
\usepackage{amsmath}
\usepackage{amstext}
\usepackage{appendix}
\usepackage{appendix}
\usepackage{mathpazo}
\usepackage{tikz}
\usepackage{verbatim}

\usetikzlibrary{trees}
\usetikzlibrary{matrix}
\usetikzlibrary{decorations.pathreplacing}
\usetikzlibrary{topaths,calc}
\usepackage{natbib}%
\providecommand{\U}[1]{\protect\rule{.1in}{.1in}}
\newtheorem{theorem}{Theorem}

\newtheorem{definition}{Definition}
\newtheorem{example}{Example}

\newtheorem{lemma}{Lemma}

\newtheorem{proposition}{Proposition}

\numberwithin{equation}{section}

\definecolor{MyDarkBlue}{rgb}{0,0.08,0.45}
\definecolor{cites}{HTML}{324b13}
\definecolor{links}{HTML}{1a663b}
\definecolor{MyLightMagenta}{cmyk}{0.1,0.8,0,0.1}
\hypersetup{
colorlinks,citecolor=blue,filecolor=black,linkcolor=blue,urlcolor=blue
}
\usepackage[top=0.9in,bottom=0.9in,left=0.9in,right=0.9in]{geometry}

\usepackage{setspace}
\usepackage{footmisc}
\usepackage{lipsum}

\begin{document}

\title{Matching to two sides}
\author{Chao Huang\thanks{Institute for Social and Economic Research, Nanjing Audit University. Email: huangchao916@163.com.}}
\date{}
\maketitle

\begin{abstract}
This paper studies a matching problem in which a group of agents cooperate with agents on two sides. In environments with either nontransferable or transferable utilities, we demonstrate that a stable outcome exists when cooperations exhibit same-side complementarity and cross-side substitutability. Our results apply to pick-side matching problems and membership competition in online duopoly markets.
\end{abstract}



\section{Introduction}\label{Sec_intro}

In real life, individuals often face choices between two competing sides. For example, 

$\bullet$ A person may join one of two rival political parties or factions.

$\bullet$ A consumer might subscribe to a membership service from one of two dominant online platforms.

$\bullet$ A nation might have to align with one of two dominant military alliances.

In some cases, an agent can only select one side. In others---such as purchasing memberships---buying both is feasible but wasteful if the services overlap. Notably, in all these scenarios, the alternatives are substitutes for the agents, yet the agents themselves are often complements for the organizations or platforms they interact with.

This paper investigates a matching market in which central agents cooperate with agents on two sides. Following the literature (e.g., \citealp{R84}; \citealp{HM05}), we model cooperation through bilateral contracts. We prove that a stable outcome exists when contracts are \textbf{same-side complementary} and \textbf{cross-side substitutable} for agents with non-transferable utilities (NTU). This framework presents the exact polar opposite of the trading-network models studied by \cite{O08} and \cite{HK12}, where contracts are same-side substitutable and cross-side complementary.

When one side of the market is empty, our model reduces to a 
two-sided many-to-many matching environment with purely complementary preferences—a special case of the framework proposed by \cite{RY20}. In that setting, a stable outcome can be derived by their one-sided Deferred Acceptance (DA) algorithm. For our setting, we show that a stable outcome can be found by running the one-sided DA procedure for the two sides alternately. 

We also show that a stable outcome exists when contracts are \textbf{same-side gross complements} and \textbf{cross-side gross substitutes} for agents with transferable utilities (TU). This setting is the polar opposite of several prior studies. \cite{SY06} examined an exchange economy with two groups of goods in which goods are within-group gross substitutes and cross-group gross complements. \cite{HKNOW13} and \cite{FJJT19} studied trading networks of sellers and buyers in which contracts are same-side gross substitutes and cross-side gross complements. We adapt the methodology of these works to our problem: We transform our market into the TU market of \cite{RY20} in which contracts are gross complements for all agents. An equilibrium of the modified market is guaranteed to exist and corresponds to an equilibrium of the original market.

Our result in the NTU market applies to a pick-side matching scenario where two competing organizations recruit from a shared pool of potential members. This framework captures real-world matching dynamics---whether in politics, the military, or business---between rival organizations and applicants. Organizations often do not view members as substitutes, as an organization barely replaces a member with another in real-life blocking processes.

Our result in the TU market applies to a duopoly membership market. For example, consider a Chinese consumer who rarely cooks and thus subscribes to a meal-delivery service from one of the two dominant platforms, Meituan or Ele.me. In this scenario, memberships function as complements rather than substitutes for the platforms, which benefit from strong economies of scale.

If each central agent is required to pick one side in the NTU market, the outcome produced by our algorithm is also setwise stable. Setwise stability---introduced by \cite{S99}---provides a criterion for many-to-many and multilateral matching that precludes blocking coalitions in which the members are more cooperative. See also \cite{EO06}, \cite{KW09}, \cite{BH21}, and \cite{H23b}, among others. 

This paper studies a matching environment incorporating both complements and substitutes, building on the studies we have discussed so far. A broad body of research addresses related topics. Stable outcomes or equilibria are guaranteed to exist in matching envrionments and exchange economies with indivisibilities under substitutability conditions that preclude complements (e.g., \citealp{KC82}; \citealp{RS90}; \citealp{GS99}; \citealp{HM05}; \citealp{FJJT19}; \citealp{PY23}). However, complements are prevalent in reality. For example, workers of different types are often complements for firms. Couples in job markets also cause complementarity (e.g., \citealp{KK05}; \citealp{KPR13}). \cite{EY07} and \cite{P12} examined complements and peer effects in job markets. 

Unit demands---an elementary type of substitutability---are incompatible with complements in settings without contract terms, as shown by \cite{GS99} and \cite{HK08}. Consequently, we can hardly expect general conditions imposed on individuals for the existence of stable outcomes/equilibria in problems like job matching or exchange economies with indivisible goods. The bilateral and unilateral substitutes conditions of \cite{HK10} relax substitutability in a setting with contracts and apply to cadet–branch matching (\citealp{S13}; \citealp{SS13}).  Some unimodularity conditions (\citealp{DKM01}; \citealp{BK19}; \citealp{H23a}) imposed on the structure of preference profiles are sufficient and allow complements. 

The literature has also explored approximate or near-feasible solutions. \cite{AWW13}, \cite{AH18}, and \cite{CKK19} demonstrated that approximate solutions exist in markets with a large number of participants. \cite{NV18,NV19,NV24} found that exact solutions exist in certain markets when parameters---such as quotas, constraints, or supplies---are slightly adjusted.

Another strand of research concerns general utility forms in economies with indivisible commodities and transfers; see \cite{FJJT19}, \cite{BEJK23}, and \cite{NV24}, among others. Our TU market result can be further extended to this framework following the analysis of \cite{RY25}.

The remainder of this paper is organized as follows. Section \ref{Sec_P} introduces the market structure and notation. Section \ref{Sec_NTU} presents the NTU model, its corresponding algorithm that finds a stable outcome, and a result on setwise stability. Section \ref{Sec_pick} applies this framework to pick-side matching. Section \ref{Sec_TU} presents the TU model and our existence result. Section \ref{Sec_membership} provides an application to membership competitions. All proofs are relegated to the Appendix.

\section{Preliminaries\label{Sec_P}}

There is a finite set $I$ of agents, which is partitioned into three subsets $I^L$, $I^M$, and $I^R$. We call agents from $I^L$, $I^M$, and $I^R$ the \textbf{left-side agents}, the \textbf{central agents}, and the \textbf{right-side agents}, respectively. There is a set of \textbf{left-side contracts} $X^L$ in which each contract $x\in X^L$ is signed by a left-side agent and a central agent, and a set of \textbf{right-side contracts} $X^R$ in which each contract $x\in X^R$ is signed by a central agent and a right-side agent. Let $X\equiv X^L\cup X^R$ be the set of all contracts. For any set of contracts $Y\subseteq X$, let $Y^L\equiv Y\cap X^L$ be the set of left-side contracts in $Y$, and let $Y^R\equiv Y\cap X^R$ be the set of right-side contracts in $Y$. 

For each contract $x\in X$, let $\mathrm{N}(x)$ be the set of the two participants of $x$. For each set of contracts $Y\subseteq X$ and each agent $i\in I$, let $Y_i\equiv\{x\in Y|i\in \mathrm{N}(x)\}$ be the subset of $Y$ in which each contract involves agent $i$. For sets of contracts $Y\subseteq X$, we write $\mathrm{N}(Y)\equiv\bigcup_{x\in Y}\mathrm{N}(x)$.

\section{Model with nontransferable utilities\label{Sec_NTU}}

\subsection{Setting}\label{Sec_setNTU}

In this section, we present the model in which agents have nontransferable utilities. In this setting, the contract set $X$ is finite. Each agent $i\in I$ has a strict preference ordering $\succ_i$ over $2^{X_i}$. For any $Y,Z\subseteq X_i$, we write $Y\succcurlyeq_iZ$ if $Y\succ_iZ$ or $Y=Z$. For each agent $i\in I$, let $\mathrm{Ch}_i:2^{X}\rightarrow 2^{X_i}$ be $i$'s \textbf{choice function} such that $\mathrm{Ch}_i(Y)\subseteq Y_i$ and $\mathrm{Ch}_i(Y)\succcurlyeq_iZ$ for each $Y\subseteq X$ and $Z\subseteq Y_i$. For any $Y\subseteq X$, let $\mathrm{Re}_i(Y)\equiv Y_i\setminus \mathrm{Ch}_i(Y)$ be agent $i$'s \textbf{rejection function}. An \textbf{outcome} in an NTU market is a set of contracts $Y\subseteq X$. 

An outcome $Y\subseteq X$ is \textbf{individually rational} for agent $i\in I$ if $Y_i=\mathrm{Ch}_i(Y)$, namely, agent $i$ does not want to unilaterally drop any contracts from $Y_i$. An outcome $Y\subseteq X$ is called individually rational if $Y_i=\mathrm{Ch}_i(Y)$ for all $i\in I$.

\begin{definition}\label{Def_NTUstable}
\normalfont
In an NTU market,
\begin{description}
\item[(\romannumeral1)] an outcome $Y\subseteq X$ is \textbf{blocked} by a nonempty $Z\subseteq X\setminus Y$ if $Z_i\subseteq \mathrm{Ch}_i(Y\cup Z)$ for all $i\in \mathrm{N}(Z)$;

\item[(\romannumeral2)] an outcome is \textbf{stable} if it is individually rational and cannot be blocked.
\end{description}
\end{definition}

A nonempty set $Z\subseteq X\setminus Y$ blocks an outcome $Y\subseteq X$ if the agents of $\mathrm{N}(Z)$ improve themselves by signing contracts of $Z$ and possibly dropping some contracts of $Y$. Notably, the newly signed contracts are in the participants' best choices from the newly signed contracts and the original contracts: $Z_i\subseteq \mathrm{Ch}_i(Y\cup Z)$ for all $i\in \mathrm{N}(Z)$.\footnote{This requirement for a block is in contrast to a notion of setwise stability; see Section \ref{Sec_setwise}.} This definition---proposed by \cite{HK12}---generalizes the stability concept in two-sided matching. It has been studied in \cite{BH21} and coincides with the one used in \cite{RY20}.

\begin{definition}
\normalfont
Contracts are \textbf{complementary} for an agent $i\in I$ if $Y\subseteq Y'\subseteq X_i$ implies 
    \begin{equation}\label{LRC}
    \mathrm{Ch}_i(Y)\subseteq \mathrm{Ch}_i(Y')
    \end{equation} 
\end{definition}

This condition requires that an agent's choice expands as more contracts become available, or equivalently, an agent does not want to substitute some contracts with any newly available contracts. In a setting that allows externalities and indifferences, \cite{RY20} showed that a stable outcome exists and can be found by a one-sided DA algorithm when contracts are complementary for all agents. We assume contracts are complementary for the left-side and right-side agents and make another assumption on central agents' preferences.

\begin{definition}\label{SCCS}
\normalfont
Contracts are \textbf{same-side complementary} and \textbf{cross-side substitutable} for a central agent $i\in I^M$ if for any $Y,Z\subseteq X_i$, $Y^L=Z^L$ and $Y^R\subseteq Z^R$ imply
\begin{equation}\label{MS1}
    {[\mathrm{Ch}_i(Y)]}^R\subseteq {[\mathrm{Ch}_i(Z)]}^R \text{ and } {[\mathrm{Ch}_i(Z)]}^L\subseteq {[\mathrm{Ch}_i(Y)]}^L,
\end{equation} 
and $Y^L\subseteq Z^L$ and $Y^R=Z^R$ imply
\begin{equation}\label{MS2}
    {[\mathrm{Ch}_i(Y)]}^L\subseteq {[\mathrm{Ch}_i(Z)]}^L \text{ and } {[\mathrm{Ch}_i(Z)]}^R\subseteq {[\mathrm{Ch}_i(Y)]}^R.
\end{equation}
\end{definition}

In words, the term (\ref{MS1}) means that if more right-side contracts become available, a central agent will choose more right-side contracts and less left-side contracts; the term (\ref{MS2}) means that if more left-side contracts become available, a central agent will choose more left-side contracts and less right-side contracts. 

When contracts are complementary for a left-side or right-side agent, we also say that contracts are same-side complementary for her, as there is only one side for her. Consequently, we can simply say that we assume contracts are same-side complementary and cross-side substitutable for all agents.

\subsection{Algorithm}\label{Sec_alg}

We describe an \textbf{alternate Deferred Acceptance} algorithm for our problem. This algorithm implements the one-sided DA (\citealp{RY20}) for the left-side market and the right-side market alternately.

As the algorithm begins, let the left-side and central agents choose from the left-side contracts. We remove the contracts rejected by anyone and let the left-side and central agents choose from the remaining contracts. We repeat this process until there are no rejections. This process is called the left-side DA. We then run the right-side DA: Let the central and right-side agents select from the contracts chosen by the left-side DA together with the right-side contracts. We remove the \emph{right-side} contracts rejected by anyone and let the central and right-side agents choose from the remaining contracts. We repeat this process until no right-side contracts are rejected, then turn to the left-side DA with the left-side contracts chosen at the end of the right-side DA. The algorithm repeats the left-side DA and the right-side DA alternately until no left-side contracts are removed by the left-side DA or the right-side DA. Formally, let $A^{1(1)}\equiv X^L$.

\bigskip

\textbf{Stage} $k, k\geq1$.

\textbf{The left-side DA, Step} $m,m\geq1$: Let the left-side and central agents choose from the set $A^{k(m)}$. If there are contracts rejected by any agent, i.e., the set 
\begin{equation*}
B^{k(m)}\equiv \bigcup_{i\in I^L\cup I^R}\mathrm{Re}_i(A^{k(m)})
\end{equation*}
is nonempty, and $A^{k(m+1)}\equiv A^{k(m)}\setminus B^{k(m)}\neq\emptyset$, then go to the next step of the left-side DA. Otherwise, the left-side DA terminates; let $A^k\equiv A^{k(m+1)}$. 

(i) If we are at Stage $k=1$, or there are contracts rejected in the left-side DA (i.e., $A^k\neq A^{k(1)})$, go to the right-side DA.

(ii) If we are at Stage $k\geq 2$, and no contracts are rejected in the left-side DA (i.e., $A^k=A^{k(1)})$, the algorithm terminates and output $A^k\cup D^{k-1}$.

\textbf{The right-side DA, Step} $n,n\geq1$: Let $D^{k(1)}\equiv X^R$ for all Stage $k$. Let the central and right-side agents choose from the set $A^k\cup D^{k(n)}$, and let
\begin{equation*}
E^{k(n)}\equiv \bigcup_{i\in I^M\cup I^R}\mathrm{Re}_i(A^k\cup D^{k(n)})
\end{equation*}
be the collection of contracts that have been rejected by some agent. If there are contracts rejected from $D^{k(n)}$, i.e., ${[E^{k(n)}]}^R\neq\emptyset$, and $D^{k(n+1)}\equiv D^{k(n)}\setminus E^{k(n)}\neq\emptyset$, then go to the next step of the right-side DA. Otherwise, the right-side DA terminates; let $A^{k+1(1)}\equiv A^k\setminus E^{k(n)}$ and $D^k\equiv D^{k(n+1)}$.

(i) If there are contracts rejected from $A^k$, i.e., ${[E^{k(n)}]}^L\neq\emptyset$, go to the left-side DA of the next stage.

(ii) If no contracts are rejected from $A^k$, i.e., ${[E^{k(n)}]}^L=\emptyset$ (and thus $E^{k(n)}=\emptyset$ since ${[E^{k(n)}]}^R=\emptyset$), the algorithm terminates and outputs $A^k\cup D^k$.

\bigskip

In this algorithm, the right-side contracts available in the first step of the right-side DA of all stages are always all of the right-side contracts of $X^R$. Moreover, in each step of the right-side DA, the left-side contracts available to central agents are holding constant (i.e., the contracts of $A^k$), thus the right-side DA is not a standard one-sided DA. 

\begin{example}\label{exam_main}
\normalfont
There are two left-side agents ($i_1^L$ and $i_2^L$), two central agents ($i_1^M$ and $i^M_2$), and two right-side agents ($i^R_1$ and $i^R_2$). There are six contracts ($x,y,z,u,v$, and $w$), which are represented by edges below. 
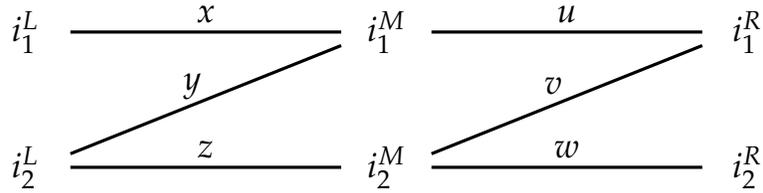
\begin{center}
\begin{tikzpicture}[scale=0.6]

  \node (n1) at (0,0) {$i^M_2$};
  \node (n2) at (0,3) {$i^M_1$};
  \node (n4) at (8,0) {$i^R_2$};
  \node (n5) at (8,3) {$i^R_1$};
  \node (n7) at (-8,0) {$i^L_2$};
  \node (n8) at (-8,3) {$i^L_1$};
  \draw [very thick](1,0)--(7,0)
        (1,0.3)--(7,2.7)
        (1,3)--(7,3)
        (-1,0)--(-7,0)
        (-1,2.7)--(-7,0.3)
        (-1,3)--(-7,3);
  \node (c1) at (-4,3.4) {$x$};
  \node (c2) at (-4.3,1.8) {$y$};
  \node (c4) at (-4,0.4) {$z$};
  \node (c6) at (4,3.4) {$u$};
  \node (c7) at (3.7,1.8) {$v$};
  \node (c8) at (4,0.4) {$w$};
\end{tikzpicture}
\end{center}
The agents have the preferences
\begin{align*}
&i^L_1: \{x\}\succ\emptyset,\qquad\qquad &i^L_2: \{y,z\}\succ\emptyset,\\
&i^M_1: \{x,y\}\succ\{u\}\succ\emptyset,\qquad\qquad &i^M_2: \{v,w\}\succ\{z\}\succ\emptyset,\\
&i^R_1: \{u,v\}\succ\{v\}\succ\emptyset,\qquad\qquad &i^R_2: \{w\}\succ\emptyset.
\end{align*}

\textbf{Stage} 1. Left-side DA: Step 1. The left-side and central agents choose from $A^{1(1)}=\{x,y,z\}$. No contract is rejected. The left-side DA terminates, and we have $A^1=\{x,y,z\}$. 

Right-side DA: Step 1. The central and right-side agents choose from $A^1\cup D^{1(1)}=\{x,y,z,u,v,w\}$. The agent $i^M_1$ rejects the contract $u$, and the agent $i^M_2$ rejects the contract $z$. We remove the rejected right-side contract $u$ from $D^{1(1)}=\{u,v,w\}$ and obtain $D^{1(2)}=\{v,w\}$.

Step 2. The central and right-side agents choose from $A^1\cup D^{1(2)}=\{x,y,z,v,w\}$. No contract of $D^{1(2)}$ is rejected, thus the right-side DA terminates, and we have $D^1=\{v,w\}$. As the contract $z$ of $A^1$ is rejected by the agent $i^M_2$, we have $A^{2(1)}=\{x,y\}$.
\medskip

\textbf{Stage} 2. Left-side DA: Step 1. The left-side and central agents choose from $A^{2(1)}=\{x,y\}$. The contract $y$ is rejected by the agent $i^L_2$. We have $A^{2(2)}=\{x\}$.

Step 2. The left-side and central agents choose from $A^{2(2)}=\{x\}$. The contract $x$ is rejected by the agent $i^M_1$. We have $A^{2(3)}=\emptyset$, thus the left-side DA terminates, and we have $A^2=\emptyset$.

Right-side DA: Step 1. The central and right-side agents choose from $A^2\cup D^{2(1)}=\{u,v,w\}$. No contract of $D^{2(1)}=\{u,v,w\}$ is rejected, thus the right-side DA terminates, and we have $D^2=\{u,v,w\}$. Since no contract of $A^2$ is rejected either, the algorithm terminates and outputs $\{u,v,w\}$.
\end{example}

Notice that the left-side contracts available for the left-side DA shrink from one stage to the next, as those left-side contracts not chosen at the end of the right-side DA become unavailable. Consequently, the set of left-side contracts that survive the left-side DA (i.e., the set $A^k$) also shrinks. Then, due to cross-side substitutability, the set of right-side contracts chosen by the right-side DA (i.e., the set $D^k$) expands from stage to stage (see Lemma \ref{lma_D} in the Appendix). Since the contracts are finite, the algorithm must eventually terminate.

Suppose the algorithm terminates at the right-side DA of Stage $s$, and the produced outcome $A^s\cup D^s$ is blocked by $Z\subseteq X\setminus(A^s\cup D^s)$. Since $A^k$ shrinks and $D^k$ expands from one stage to the next, the available left-side contracts in the algorithm always contain those of $A^s$, and the available right-side contracts in the last step of the right-side DA of each Stage (which are those of $D^k$) are always contained in $D^s$. Then, since $Z_i\subseteq\mathrm{Ch}_i(A^s\cup D^s\cup Z)$ for all $i\in \mathrm{N}(Z)$, those left-side contracts of $Z^L$ cannot be removed in the algorithm due to same-side complementarity and cross-side substitutability. Consequently, we have $Z^L=\emptyset$. In the final stage, given the left-side contracts of $A^s$ chosen in the left-side DA, agents do not regret rejecting any right-side contracts during the right-side DA; we thus have $Z^R=\emptyset$. We can achieve the same conclusion if the algorithm terminates after a left-side DA.

\begin{theorem}\label{thm_non}
\normalfont
The alternate DA algorithm finds a stable outcome when contracts are same-side complementary and cross-side substitutable for all agents.
\end{theorem}

Since the left side and the right side are symmetric, one can swap the two sides in the algorithm to obtain a possibly different stable outcome. Swapping the two sides does not change the output in Example \ref{exam_main} but produces a different stable outcome in Example \ref{exam_org} of the Appendix.

\subsection{Setwise stability}\label{Sec_setwise}

Stable outcomes under Definition \ref{Def_NTUstable} are not immune to a certain type of renegotiations, which we show by an example below. In this section, we introduce the concept of setwise stability, which prevents such renegotiations. We then show that stability and setwise stability coincide when we further require that each central agent can only cooperate with agents on one side.

Consider a market with one left-side agent $i^L$, one central agent $i^M$, and one right-side agent $i^R$. The agents $i^L$ and $i^M$ can sign contracts $x$ and $y$; and the agents $i^M$ and $i^R$ can sign contract $z$. The agents have the preferences
\begin{equation}\label{exam_setwise}
\begin{aligned}
&i^L: \{x,y\}\succ\emptyset,\\
&i^M: \{x,z\}\succ\{x,y\}\succ\{x\}\succ\{z\}\succ\emptyset,\\
&i^R: \{z\}\succ\emptyset.
\end{aligned}
\end{equation}

The only substitute in this market is the replacement of the contract $y$ with the contract $z$ by the central agent $i^M$ when she also holds the contract $x$. The contracts are same-side complementary and cross-side substitutable for all agents. There exists a unique stable outcome $\{z\}$. However, when the central agent $i^M$ holds the contract $z$, she would rather sign the contracts $x$ and $y$ with the left-side agent $i^L$ and drop $z$. Notice that this is not a block under Definition \ref{Def_NTUstable}: The set $\{x,y\}$ is not $i^M$'s best choice from the original contract ($z$) and the newly signed contracts ($x$ and $y$). Therefore, in the renegotiation between $i^L$ and $i^M$, $i^M$ should promise $i^L$ not to drop any of the newly signed contracts. The notion of setwise stability proposed by \cite{S99} prevents such renegotiations.
\begin{definition}\label{Def_setwise}
\normalfont
In an NTU market,
\begin{description}
\item[(\romannumeral1)] an outcome $Y\subseteq X$ is \textbf{setwise blocked} by a nonempty $Z\subseteq X\setminus Y$ if there exists an outcome $Y^*\subseteq Y\cup Z$ such that $Z\subseteq Y^*$, $ Y^*_i\succ_iY_i$, and $Y^*$ is individually rational for all $i\in \mathrm{N}(Z)$;

\item[(\romannumeral2)] an outcome is \textbf{setwise stable} if it is individually rational and cannot be setwise blocked.
\end{description}
\end{definition}
A group of agents can implement a setwise block if they can renegotiate to a new outcome that is better and individually rational for all participants. In this new outcome, the agents involved in the renegotiation may not obtain their best choices from the newly signed contracts and the original contracts, as illustrated by market (\ref{exam_setwise}). The outcome $\{z\}$ of this market is setwise blocked by $\{x,y\}$. Setwise stability has been studied by \cite{EO06} and \cite{KW09} in many-to-many matching and by \cite{BH21} in multilateral matching. There is another distinction between the block requirements of the two stability concepts. When an outcome $Y$ is setwise blocked by $Z$, the agents of $\mathrm{N}(Z)$ should make consistent decisions on which contracts to drop such that the original outcome is brought into a new outcome. However, a block under Definition \ref{Def_NTUstable} may not satisfy this requirement, as shown by the following market.\footnote{This is a market in which every pair of agents can sign contracts.}
\begin{equation}\label{exam_stable}
i_1: \{x,y\}\succ\{x\}\succ\emptyset \qquad\qquad i_2: \{z\}\succ\{x\}\succ\emptyset \qquad\qquad i_3: \{y,z\}\succ\emptyset
\end{equation}
Suppose the agents $i_1$ and $i_2$ have signed the contract $x$; then, both agents want to sign a new contract with the agent $i_3$. However, after signing the new contracts, $i_2$ wishes to drop $x$, while $i_1$ does not. According to Definition \ref{Def_NTUstable}, the outcome $\{x\}$ is blocked by $\{y,z\}$. This is a block in which the participants make inconsistent decisions about which contracts to drop. However, the outcome $\{x\}$ is setwise stable, as setwise stability does not rule out such blocks.

From the two examples above, we can conclude that stability and setwise stability are independent concepts. Stability is a suitable solution criterion for economies in which agents make decisions independently, whereas setwise stability is better suited for environments in which agents are more cooperative.

No setwise stable outcome exists in market (\ref{exam_setwise}), demonstrating that a setwise stable outcome is not guaranteed even under conditions of same-side complementarity and cross-side substitutability. However, stability and setwise stability coincide when we further require that each central agent can only cooperate with agents on one side. We say that a central agent $i\in I^M$ \textbf{has to pick one side} if for any $Y\subseteq X_i$, $Y\succ_i\emptyset$ implies $Y\subseteq X^L_i$ or $Y\subseteq X^R_i$.

\begin{proposition}\label{prop_equi}
\normalfont
When contracts are same-side complementary for all agents\footnote{Contracts are same-side complementary for a central agent if the condition of Definition \ref{SCCS} only holds for the ``${[\mathrm{Ch}_i(Y)]}^R\subseteq {[\mathrm{Ch}_i(Z)]}^R $'' part of (\ref{MS1}) and the ``${[\mathrm{Ch}_i(Y)]}^L\subseteq {[\mathrm{Ch}_i(Z)]}^L $'' part of (\ref{MS2}).\label{foot_SC}} and each central agent has to pick one side, an outcome is stable if and only if it is setwise stable.
\end{proposition}

This result and Theorem \ref{thm_non} imply that a stable outcome---which is also setwise stable---exists when contracts are same-side complementary for all agents and each central agent has to pick one side.

\section{Application: Pick-side matching}\label{Sec_pick}

A direct application of our result in the NTU market is a market in which agents choose to join one of two competing organizations---a scenario prevalent in politics, commerce, and military affairs. This problem is a special many-to-one matching market with two organizations on one side. It is well-established that a stable matching exists in such markets if applicants are substitutes for organizations (see \citealp{RS90}).  However, applicants are often complements rather than substitutes for organizations due to the following reasons:

(i) Students are usually substitutes for schools due to strict quotas. Organizations, by contrast, typically operate without such limitations.

(ii) The fundamental purpose of an organization is to combine individuals to leverage their complementary effects and interactions, which are at odds with substitutability. In particular, in real-life blocks, it is common for a firm to substitute one worker with another; however, it is rare that an organization replaces one member with another.

(iii) When two workers are qualified for the same position, they are substitutes for a firm, which will not hire both if one is sufficient. This is because a firm incurs a direct cost for each employee. Organizations, however, face no such salary constraints.

Concrete environments also exist in which members are complementary for organizations. In practice, organizations often establish entry standards for potential applicants. For example, Ukraine's application to join NATO was declined during the Russia-Ukraine conflict due to Article 10 of the NATO Treaty, which requires applicant states to resolve major territorial disputes before joining the alliance. Meanwhile, Finland and Sweden were admitted as they met all entry requirements. In this case, if NATO evaluates applicants based on individual merit without raising its standards as more applicants emerge, then applicants are not substitutes for NATO.

Another key consideration is member connectivity. Individuals form organizations to enhance coordination and interaction. Consequently, organizations naturally evaluate applicants based on their connectivity to existing or potential members, which creates a fundamental source of complementarity. For instance, a commercial organization might decline to recruit firm A due to its limited business connections with current members. However, the same organization might hire both firms A and B if B has strong ties to current members and A, in turn, has strong connections with B.

The alternate DA algorithm can be adapted into the following intuitive procedure in the matching between two organizations ($o_1$ and $o_2$) and potential new members.

\bigskip

The agents who accept the organization $o_1$ are initially matching to $o_1$; other agents are unmatched.

\textbf{Stage $k$, $k\geq1$}. The agents currently matching to $o_1$ propose to $o_1$; and $o_1$ chooses its favorite set from the proposals and rejects the rest. The algorithm terminates and output the two organizations' current choices if we are at Stage $k\geq 2$ and no agents are rejected by $o_1$. Otherwise, the agents who weakly prefer $o_2$ to their current positions propose to $o_2$; and $o_2$ chooses its favorite set from the proposals and rejects the rest. The algorithm terminates and output the two organizations' current choices if the agents chosen by $o_2$ contain no agent from $o_1$. Otherwise, go to the next stage.

\bigskip

When agents are complementary for the two organizations, the proposals to organization $o_1$ shrink while those to organization $o_2$ expand within the algorithm. We provide an example of the alternate DA algorithm for this market in Section \ref{Sec_exam} of the Appendix.

The following example, which involves three organizations ($o_1, o_2$, and $o_3$) and three agents ($i_1, i_2$, and $i_3$), demonstrates that our results for the NTU market do not extend to scenarios in which agents can choose among three or more sides. For this example, we adopt the notation of \cite{RS90}.

\begin{equation}\label{exam_in1}
\begin{aligned}
&o_1: \{i_1,i_2\}\succ\emptyset \qquad\qquad\qquad\qquad &i_1: &\quad o_1\succ o_3\\
&o_2: \{i_2,i_3\}\succ\emptyset \qquad\qquad\qquad\qquad &i_2: &\quad o_2\succ o_1\\
&o_3: \{i_1,i_3\}\succ\emptyset \qquad\qquad\qquad\qquad &i_3: &\quad o_3\succ o_2\\
\end{aligned}
\end{equation}
If the agents $i_1$ and $i_2$ join the organization $o_1$, leaving the agent $i_3$ unmatched, then $i_2$ would rather join the organization $o_2$ with $i_3$. Similarly, neither $\{i_2,i_3\}$ matching to $o_2$ nor $\{i_1,i_3\}$ matching to $o_3$ is stable.

\section{Model with transferable utilities}\label{Sec_TU}

\subsection{Setting}\label{Sec_setTU}

In this section, we present the model in which agents have transferable utilities. In this setting, each contract $x=(w,\mathbf{p}^w)\in X$ contains a \textbf{primitive contract} $w\in\Omega$ and a price vector $\mathbf{p}^w=(p^w_i,p^w_j)\in \mathbb{R}^{\mathrm{N}(x)}$ with $p^w_i+p^w_j=0$ where $\{i,j\}=\mathrm{N}(x)$. The primitive contract $w$ describes the contract's nonpecuniary part; and the price vector $\mathbf{p}^w=(p^w_i,p^w_j)$ specifies the monetary transfer from $i$ to $j$ (if $p^w_i\geq 0$) or from $j$ to $i$ (if $p^w_j\geq 0$). The set of primitive contracts $\Omega$ is finite. The primitive contract of a contract $x\in X$ is denoted as $\tau(x)$. For each set of contracts $Y\subseteq X$, the set of primitive contracts involved in $Y$ is also denoted as $\tau(Y)\equiv\{\tau(y)|y\in Y\}$. For each contract $x\in X$, its primitive contract $\tau(x)$ is associated with the set of its two participants $\mathrm{N}(\tau(x))\equiv \mathrm{N}(x)$. For each agent $i\in I$ and each set of primitive contracts $\Psi\subseteq\Omega$, we write $\Psi_i\equiv\{w\in \Psi|i\in \mathrm{N}(w)\}$ as the set of primitive contracts associated with agent $i$.

Each agent $i\in I$ has a \textbf{valuation} $\mathrm{v}_i:2^{\Omega_i}\rightarrow \mathbb{R}$ over subsets of primitive contracts signed by her. Agents' preferences over sets of contracts are quasilinear in transfers, and agents do not sign a primitive contract twice. Formally, each agent $i$'s utility of signing contracts from $Y\subseteq X_i$ is given by
\begin{equation*}
\mathrm{u}_i(Y)\equiv\left\{
\begin{aligned}
&\mathrm{v}_i(\tau(Y))-\sum_{(w,\mathbf{p}^w)\in Y}p^w_i,\quad &\text{ if }\tau(x)\neq\tau(x') \text{ for all } x,x'\in Y,\\
&-\infty,\quad &\text{otherwise}.
\end{aligned}
\right.
\end{equation*}
Agent $i$'s utility function induces her choice correspondence $\mathrm{Ch}_i : 2^X\rightrightarrows2^{X_i}$ defined by $\mathrm{Ch}_i(Y)\equiv\arg\max_S\{\mathrm{u}_i(S) \text{ s.t. } S\subseteq Y_i\}$. 

An \textbf{outcome} in a TU market is a set of contracts $Y\subseteq X$ in which different contracts are associated with different primitive contracts: $\tau(x)\neq\tau(x')$ for all $x,x'\in Y$. 
\begin{definition}\label{def_stable}
\normalfont
An outcome $Y\subseteq X$ in a TU market is \textbf{stable} if it is 
\begin{description}
  \item[(i)] \textbf{individually rational}: $Y_i\in \mathrm{Ch}_i(Y)$ for all $i\in I$, and
  \item[(ii)] \textbf{unblocked}: There is no $Z\subseteq X\setminus Y$ such that for all $i\in \mathrm{N}(Z)$, we have $Z_i\subseteq A^i$ and $\mathrm{u}_i(A^i)>\mathrm{u}_i(Y_i)$ for some $A^i\in \mathrm{Ch}_i(Y\cup Z)$.
\end{description}
\end{definition}
A set $Z$ blocks $Y$ if any agent $i\in \mathrm{N}(Z)$ holding the contracts of $Y_i$ is willing to sign all contracts from $Z_i$ while possibly dropping some contracts of $Y_i$. In particular, each participant of the block becomes strictly better off.\footnote{Since we assume transferable utilities, the stability concept is unaffected if we instead preclude blocks in which each participant is weakly better off and at least one of them becomes strictly better off.} The stability concept of Definition \ref{def_stable} is stronger than the one used in \cite{HKNOW13} and weaker than the one used in \cite{RY20}. The former's block concept requires $Z_i\subseteq A^i$ for all $A^i\in \mathrm{Ch}_i(Y\cup Z)$, which implies that each participant becomes strictly better off. The latter's does not require each participant to become strictly better off.

Let $E\equiv\{(w,i)|w\in\Omega\text{ and }i\in N(w)\}$ be the set of all primitive contract-participant pair. We use a price vector $\mathbf{p}$ in
\begin{equation*}
\mathbb{B}\equiv\{\mathbf{p}\in\mathbb{R}^E | p^w_i+p^w_j=0 \text{ with } i,j\in \mathrm{N}(w) \text{ for each } w\in \Omega\} 
\end{equation*}
to specify the prices in all primitive contracts. Given a price vector $\mathbf{p}\in\mathbb{B}$, we use $\mathbf{p}_i=(p^w_i)_{w\in\Omega_i}$ to denote the components of $\mathbf{p}$ associated with agent $i\in I$, and we also use $\mathbf{p}^w=(p^w_i)_{i\in\mathrm{N}(w)}$ to denote the components of $\mathbf{p}$ associated with primitive contract $w\in\Omega$. Agent $i$'s \textbf{demand correspondence} is given by
\begin{equation*}
\mathrm{D}_i(\textbf{p}_i)\equiv\arg\max_{\Psi\in\Omega_i}\{\mathrm{v}_i(\Psi)-\sum_{w\in\Psi}p^w_i\}.
\end{equation*}
A set of primitive contracts $\Phi\subseteq\Omega$ and a price vector $\mathbf{p}\in\mathbb{B}$ constitute a \textbf{competitive equilibrium} $(\Phi,\mathbf{p})$ if $\Phi_i\in \mathrm{D}_i(\mathbf{p}_i)$ for each $i\in I$. 
For any set of primitive contracts $\Phi\subseteq\Omega$ under price vector $\mathbf{p}\in\mathbb{B}$, let $\kappa(\Phi,\mathbf{p})\equiv\{(w,\mathbf{p}^w)|w\in\Phi\}$ be the corresponding set of contracts. We can use competitive equilibrium as an intermediate tool for studying stable outcomes.

\begin{lemma}\label{lma_stable}
\normalfont
If $(\Phi,\mathbf{p})$ is a competitive equilibrium, then $\kappa(\Phi,\mathbf{p})$ is a stable outcome.
\end{lemma}

The converse of this statement is not true: Given a stable outcome $Y$, there does not necessarily exists a price vector $\mathbf{p}\in\mathbb{B}$ consistent with $Y$ (i.e., $(w,\mathbf{p}^w)\in Y$ if $w\in\tau(Y)$) such that $(\tau(Y),\mathbf{p})$ is a competitive equilibrium; see Example 1 of \cite{HKNOW13}. 

\cite{RY20} showed that competitive equilibria exists under a gross complementarity condition that is equivalent to supermodularity.

\begin{definition}\label{def_GC}
\normalfont
Primitive contracts are \textbf{gross complements} for agent $i\in I$ if for any $\mathbf{p}_i\geq \mathbf{q}_i\in \mathbb{R}^{\Omega_i}$ such that $|\mathrm{D}_i(\mathbf{p}_i)|=|\mathrm{D}_i(\mathbf{q}_i)|=1$, for the unique $\Phi\in \mathrm{D}_i(\mathbf{p}_i)$ and $\Psi\in \mathrm{D}_i(\mathbf{q}_i)$, we have $\{w\in\Phi|p^w_i=q^w_i\}\subseteq\Psi$.
\end{definition}

Roughly speaking, gross complementarity means that the fall of prices of some primitive contracts does not decrease the demand for other primitive contracts. The following lemma summarizes conditions equivalent to this definition.

\begin{lemma}\label{lma_RGC}
\normalfont
The following statements are equivalent.
\begin{description}
  \item[(i)] Agent $i$'s valuation $\mathrm{v}_i$ is \textbf{supermodular}: $\mathrm{v}_i(\Phi)+\mathrm{v}_i(\Psi)\leq \mathrm{v}_i(\Phi\cup \Psi)+\mathrm{v}_i(\Phi\cap \Psi)$ for any $\Phi,\Psi\subseteq \Omega_i$.
  \item[(ii)] Agent $i$'s demand correspondence $\mathrm{D}_i$ is \textbf{antitone}: for any price vectors $\mathbf{p}_i\geq \mathbf{q}_i\in \mathbb{R}^{\Omega_i}$, $\Phi\in \mathrm{D}_i(\mathbf{p}_i)$, and $\Psi\in \mathrm{D}_i(\mathbf{q}_i)$, we have $\Phi\cap\Psi\in \mathrm{D}_i(\mathbf{p}_i)$ and $\Phi\cup\Psi\in \mathrm{D}_i(\mathbf{q}_i)$.
  \item[(iii)] For any $\mathbf{p}_i\geq \mathbf{q}_i\in \mathbb{R}^{\Omega_i}$ and $\Phi\in \mathrm{D}_i(\mathbf{p}_i)$ of agent $i\in I$, there exists $\Psi\in \mathrm{D}_i(\mathbf{q}_i)$ such that $\{w\in\Phi|p^w_i=q^w_i\}\subseteq\Psi$.      
  \item[(iv)] Contracts are gross complements for agent $i\in I$. 
  \item[(v)] For any $\mathbf{p}_i\geq \mathbf{q}_i\in \mathbb{R}^{\Omega_i}$ and $\Psi\in \mathrm{D}_i(\mathbf{q}_i)$ of agent $i\in I$, there exists $\Phi\in \mathrm{D}_i(\mathbf{p}_i)$ such that $\{w\in\Phi|p^w_i=q^w_i\}\subseteq\Psi$.
\end{description}
\end{lemma}

\citet[Lemma 2]{RY20} showed (i)$\Leftrightarrow$(ii), and \citet[Theorem 1]{Y23} showed (ii)$\Leftrightarrow$(iii).
Conditions (iii) and (v) are symmetric, and each of them immediately implies (iv). It is also obvious that (iii) and (v) follow from (ii): For any $\mathbf{p}_i\geq \mathbf{q}_i\in \mathbb{R}^{\Omega_i}$, $\Phi\in \mathrm{D}_i(\mathbf{p}_i)$, and $\Psi\in \mathrm{D}_i(\mathbf{q}_i)$, we have $\Phi\cup\Psi\in \mathrm{D}_i(\mathbf{q}_i)$ satisfying $\{w\in\Phi|p^w_i=q^w_i\}\subseteq\Phi\cup\Psi$, and we have $\Phi\cap\Psi\in \mathrm{D}_i(\mathbf{p}_i)$ satisfying $\{w\in\Phi\cap\Psi|p^w_i=q^w_i\}\subseteq\Psi$. We complete the proof by showing (iv)$\Rightarrow$(i) in the Appendix.

We assume contracts are gross complements for all left-side and right-side agents. We make another assumption on central agents' preferences.

\begin{definition}\label{def_SSGC}
\normalfont
Primitive contracts are \textbf{same-side gross complements} and \textbf{cross-side gross substitutes} for a central agent $i\in I^M$ if 
\begin{description}
  \item[(i)]  for any price $\mathbf{p}_i,\mathbf{q}_i\in \mathbb{R}^{\Omega_i}$ such that such that $|\mathrm{D}_i(\mathbf{p}_i)|=|\mathrm{D}_i(\mathbf{q}_i)|=1$, $p^w_i=q^w_i$ for all $w\in\Omega_i^L$, and $p^w_i\geq q^w_i$ for all $w\in\Omega_i^R$, for the unique $\Phi\in \mathrm{D}_i(\mathbf{p}_i)$ and $\Psi\in \mathrm{D}_i(\mathbf{q}_i)$, we have $\{w\in\Phi^R|p^w_i=q^w_i\}\subseteq\Psi^R$ and $\Psi^L\subseteq\Phi^L$; and
    \item[(ii)]  for any price $\mathbf{p}_i,\mathbf{q}_i\in \mathbb{R}^{\Omega_i}$ such that such that $|\mathrm{D}_i(\mathbf{p}_i)|=|\mathrm{D}_i(\mathbf{q}_i)|=1$, $p^w_i=q^w_i$ for all $w\in\Omega_i^R$, and $p^w_i\geq q^w_i$ for all $w\in\Omega_i^L$, for the unique $\Phi\in \mathrm{D}_i(\mathbf{p}_i)$ and $\Psi\in \mathrm{D}_i(\mathbf{q}_i)$, we have $\{w\in\Phi^L|p^w_i=q^w_i\}\subseteq\Psi^L$ and $\Psi^R\subseteq\Phi^R$. 
\end{description}
\end{definition}

This condition states that for a central agent, a price decrease for some primitive contracts on one side will not reduce her demand for other primitive contracts on the same side, but will reduce her demand for primitive contracts on the opposite side. One can also define this condition as different forms of full substitutability given by \citet[Theorem A.1]{HKNOW19}. We assume contracts are same-side gross complements and cross-side gross substitutes for all central agents. 

Again, since agents on the left or right side only interact with one side of the market, we say that contracts are same-side gross complements for a left-side or right-side agent if contracts are gross complements for her. Consequently, we can simply say that we assume contracts are same-side gross complements and cross-side gross substitutes for all agents. In line with the literature, this assumption may also be termed \textbf{full complementarity}. 

\subsection{Transformation}\label{Sec_tran}

We prove the existence of a stable outcome following the method used in \cite{SY06}, \cite{HKNOW13}, and \cite{FJJT19}. We transform the original market into a new one by changing the valuation of each central agent $i\in I^M$ into
\begin{equation*}
\mathrm{\tilde{v}}_i(\Psi)\equiv\mathrm{v}_i(\Psi^L\cup(\Omega_i^R\setminus \Psi^R)).
\end{equation*}
for each $\Psi\subseteq \Omega_i$ and changing the valuation of each right-side agent $j\in I^R$ into
\begin{equation*}
\mathrm{\tilde{v}}_j(\Phi)\equiv\mathrm{v}_j(\Omega_j\setminus \Phi).
\end{equation*}
for each $\Phi\subseteq \Omega_j$. A central agent's value on a set of primitive contracts is transformed into her value on the set of the left-side primitive contracts she signs and the right-side primitive contracts she does not sign. A right-side agent's value on a set of primitive contracts is transformed into her value on the set of the primitive contracts she does not sign.

\begin{lemma}\label{lma_condition}
\normalfont
\begin{description}
  \item[(i)] If primitive contracts are gross complements for a right-side agent $i\in I^R$ in the original market, then primitive contracts are gross complements for $i$ in the modified market.
  \item[(ii)] If primitive contracts are same-side gross complements and cross-side gross substitutes for a central agent $i\in I^M$ in the original market, then primitive contracts are gross complements for $i$ in the modified market.  
\end{description}
\end{lemma}

Let $\mathrm{g}$ be an operator on a price vector in $\mathbb{B}$ that reverses the directions of the transfers of all right-side contracts: For every $\mathbf{p}\in \mathbb{B}$, let $\mathrm{g}(\mathbf{p})\in\mathbb{B}$ be the price vector whose $(w,i)$-wise component is $p_i^w$ if $w\in \Omega^L$ and $-p^w_i$ if $w\in \Omega^R$.

\begin{lemma}\label{lma_CE}
\normalfont
If $(\Psi,\mathbf{p})$ is a competitive equilibrium in the modified market, then $(\Psi^L\cup(\Omega^R\setminus\Psi^R),\mathrm{g}(\mathbf{p}))$ is a competitive equilibrium in the original market. 
\end{lemma}

Lemma \ref{lma_condition} and an existence result of \citet[Proposition 3]{RY20} implies the existence of a competitive equilibrium in our modified market; then, Lemma \ref{lma_CE} implies the existence of a competitive equilibrium in the original market. Consequently, we obtain the following result according to Lemma \ref{lma_stable}. 

\begin{theorem}\label{thm_TU}
\normalfont
A stable outcome exists in a TU market if contracts are same-side gross complements and cross-side gross substitutes for all agents.  
\end{theorem}

\section{Application: Membership competition}\label{Sec_membership}

Online start-ups often emerge rapidly across many sectors,  but most ultimately falter in the competition. As firms merge or exit, many sectors evolve into oligopoly markets. Our finding for the TU market applies to an online duopoly in which two companies sell memberships to consumers. Unlike offline markets, online prices frequently vary over time and across different consumers.

For instance, a Chinese professor who often teaches online might subscribe to one of the two dominant virtual-meeting platforms, Tencent Meeting or DingTalk. Similarly, in the U.S., the choices might be Zoom and Microsoft Teams. When a customer buys a membership from one company, she gains access to upgraded services or discounted rates for individual services. Since memberships from different providers offer similar benefits, customers typically do not subscribe to more than one.

This dynamic extends to other online sectors, such as food delivery, ride-hailing, cloud storage, and cloud-server hosting. Memberships acquired by different customers often function as complements for a firm, as scale effects tend to outweigh substitution effects in internet-based industries. Moreover, in industries like cloud storage and cloud-server hosting, customers are often enterprises. In these scenarios, service providers usually charge personalized prices---an aspect our model explicitly accommodates.

\section{Appendix}

\subsection{Proofs for Section \ref{Sec_NTU}}\label{proof_NTU}

\begin{proof}[Proof of Theorem \ref{thm_non}]
We first show that at each Sage $k$ of the algorithm, the set $A^k$ produced by the left-side DA is the largest subset of $A^{k(1)}$ that is individually rational for all left-side and central agents, and the set $D^k$ produced by the right-side DA is the largest subset of $X^R$ among those that would be chosen by all central and right-side agents in the presence of $A^k$.
\begin{lemma}\label{Lma_IR}
\normalfont
Suppose contracts are same-side complementary\footnote{See footnote \ref{foot_SC} for the definition of same-side complementarity.} for all agents. At each Stage $k$ of the algorithm, 
\begin{description}
  \item[(i)] if $Y\subseteq A^{k(1)}$ satisfies $\mathrm{Ch}_i(Y)=Y_i$ for all $i\in I^L\cup I^M$, then $Y\subseteq A^k$; and
  \item[(ii)] if Stage $k$ contains a right-side DA, and $Y\subseteq X^R$ satisfies $Y_i\subseteq\mathrm{Ch}_i(A^k\cup Y)$ for all $i\in I^M\cup I^R$, then $Y\subseteq D^k$.
\end{description}
\end{lemma}

\begin{proof}
(i) Let $Y\subseteq A^{k(1)}$ be a set of contracts satisfying $\mathrm{Ch}_i(Y)=Y_i$ for all $i\in I^L\cup I^M$. Since contracts are same-side complementary, $x\in Y_i=\mathrm{Ch}_i(Y)$ implies $x\in \mathrm{Ch}_i(A^{k(1)})$ for all $x\in Y$ and $i\in \mathrm{N}(x)$. Thus, no element of $Y$ is rejected in the first step of the left-side DA; we have $Y\subseteq A^{k(2)}$. This argument applies to each step of the left-side DA, thus no element of $Y$ is rejected in any step of the left-side DA. Hence, we have $Y\subseteq A^k$.

(ii) Let $Y\subseteq X^R$ be a set satisfying $Y_i\subseteq\mathrm{Ch}_i(A^k\cup Y)$ for all $i\in I^M\cup I^R$. Since contracts are same-side complementary, $x\in Y_i\subseteq\mathrm{Ch}_i(A^k\cup Y)$ implies $x\in \mathrm{Ch}_i(A^k\cup X^R)$ for all $x\in Y$ and $i\in \mathrm{N}(x)$. Thus, no element of $Y$ is rejected in the first step of the right-side DA; and we have $Y\subseteq D^{k(2)}$. Since the left-side contracts available in the following steps are always those of $A^k$, the above argument applies to each step of the right-side DA, and thus we know that no element of $Y$ is rejected in any step of the right-side DA. Hence, we have $Y\subseteq D^k$.
\end{proof}

Suppose the algorithm terminates at Stage $s$. Since the set $A^k$ is from $A^{k(1)}$ via the left-side DA, and $A^{k(1)}$ is chosen from $A^{k-1}$ at the last step of the right-side DA, we have $A^s\subseteq A^{s(1)}\subseteq A^{s-1}\subseteq A^{s-1(1)}\subseteq\cdots \subseteq A^1\subseteq A^{1(1)}$. The following lemma shows the expansion of the right-side contracts chosen by the right-side DA.

\begin{lemma}\label{lma_D}
\normalfont
Let $s'$ be the last stage that includes a right-side DA. We have $D^k\subseteq D^{k+1}$ for all $k\in\{1,\cdots,s'-1\}$ in the algorithm if contracts are same-side complementary and cross-side substitutable for all agents.
\end{lemma}

\begin{proof}
At the last step of the right-side DA of Stage $k\in\{1,\cdots,s'-1\}$ of the algorithm, we have $D_i^k\subseteq\mathrm{Ch}_i(A^k\cup D^k)$ for all $i\in I^M\cup I^R$. Since $A^{k+1}\subseteq A^k$ and contracts are cross-side substitutable for central agents, we have $D_i^k\subseteq\mathrm{Ch}_i(A^{k+1}\cup D^k)$ for all $i\in I^M\cup I^R$. Then, according to Lemma \ref{Lma_IR} (ii), we have $D^k\subseteq D^{k+1}$.
\end{proof}

We have shown that $A^k$ shrinks and $D^k$ expands within the algorithm. The algorithm must eventually terminate since the contracts are finite. 

Now we modify the algorithm by preventing termination at left-side DA of any stage. Instead, the algorithm is required to proceed to the right-side DA in each stage. Consequently, if the original algorithm terminates at the right-side DA of some stage, the modified algorithm will also terminate there and produce an identical outcome. The following lemma demonstrates that, if the original algorithm terminates at the left-side DA of some stage, the modified algorithm terminates at the right-side DA of that stage and yields the same outcome.

\begin{lemma}\label{lma_extend}
\normalfont
Suppose contracts are cross-side substitutable\footnote{Contracts are cross-side substitutable for a central agent if the condition of Definition \ref{SCCS} only holds for the ``${[\mathrm{Ch}_i(Z)]}^L\subseteq {[\mathrm{Ch}_i(Y)]}^L $'' part of (\ref{MS1}) and the ``${[\mathrm{Ch}_i(Z)]}^R\subseteq {[\mathrm{Ch}_i(Y)]}^R $'' part of (\ref{MS2}).} for each central agent. If the original algorithm terminates at the left-side DA of Stage $s\geq 2$, then the last step of the right-side DA of Stage $s$ of the modified algorithm chooses $A^s\cup D^s=A^s\cup D^{s-1}$, and then the modified algorithm terminates.
\end{lemma}
\begin{proof}
In the right-side DA, as right-side contracts are removed in each step, the chosen left-side contracts of each central agent in each step of the right-side DA expand due to cross-side substitutability. At the right-side DA of Stage $s-1$, for each central agent $i\in I^M$, the left-side contracts of $A_i^{s-1}$ are available in each step, and the left-side contracts of $A_i^{s(1)}=A_i^s\subseteq A_i^{s-1}$ are chosen at the last step,\footnote{Notice that, in the right-side DA, a right-side contract may be rejected by a central agent or a right-side agent, but a left-side contract can only be chosen or rejected by the central agent who signs this contract.} where the equality is due the termination of the original algorithm at Stage $s$. Hence, for each central agent $i\in I^M$, the left-side contracts chosen in each step of the left-side DA of Stage $s-1$ belong to $A_i^s$.

Compare the first step of the right-side DA in Stage $s-1$ and Stage $s$: for each central agent $i\in I^M$, the left-side contracts available are from $A_i^{s-1}$ in the former and from $A_i^s\subseteq A_i^{s-1}$ in the latter; the right-side contracts available are both from $X_i^R$. Since the left-side contracts chosen by each central agent $i\in I^M$ in the left-side DA of Stage $s-1$ belong to $A_i^s$, each central agent's optimal choice in the first step of the right-side DA of Stage $s-1$ is also available in the first step of the right-side DA of Stage $s$; thus each central agent chooses the same contracts in the first step of the right-side DA of Stage $s-1$ and Stage $s$. Hence, the right-side contracts removed are also the same. This argument applies to each step of the right-side DA of Stage $s-1$ and Stage $s$. Consequently, the last step of the right-side DA of Stage $s$ of the modified algorithm chooses $A^s\cup D^s=A^s\cup D^{s-1}$, and the modified algorithm terminates.
\end{proof}

Due to this lemma, we prove the theorem by showing that the outcome produced by the modified algorithm is stable. Notice that Lemma \ref{Lma_IR} and Lemma \ref{lma_D} also hold for the modified algorithm. Suppose the modified algorithm terminates at Stage $s$ and produces $A^s\cup D^s$.

The produced outcome $A^s\cup D^s$ is individually rational for the left-side agents according to the definition of $A^s$. It is also individually rational for the central and right-side agents due to condition (ii) in the right-side DA for the termination of the algorithm. 

Suppose the produced outcome $A^s\cup D^s$ is blocked by $Z\subseteq X\setminus(A^s\cup D^s)$. Since $A^k$ shrinks from one stage to the next, the available left-side contracts in any left-side DA or right-side DA always contain those of $A^s$. At each Stage $k$, the right-side contracts of $D^k$ are chosen at the last step of the right-side DA, and no right-side contracts are rejected at this step. Hence, the available right-side contracts at the last step of the right-side DA of each Stage $k$ are also those of $D^k$, which are contained in $D^s$ (due to Lemma \ref{lma_D}). Then, since $Z_i\subseteq\mathrm{Ch}_i(A^s\cup D^s\cup Z)$ for all $i\in \mathrm{N}(Z)$, those left-side contracts of $Z^L$ cannot be removed in the algorithm due to same-side complementarity and cross-side substitutability. Consequently, we have $Z^L=\emptyset$. 

Now we know $Z\subseteq X^R$, and thus $D^s\cup Z\subseteq X^R$.
We have $D^s_i\subseteq\mathrm{Ch}_i(A^s\cup D^s)$ for all $i\in I^M\cup I^R$ at the final step of the right-side DA of Stage $s$ of the modified algorithm. Then, since contracts are same-side complementary,  $Z_i\subseteq\mathrm{Ch}_i(A^s\cup D^s\cup Z)$ for all $i\in \mathrm{N}(Z)$, and $D^s\cup Z\subseteq X^R$, we have $D^s_i\cup Z_i\subseteq\mathrm{Ch}_i(A^s\cup D^s\cup Z)$ for all $i\in \mathrm{N}(Z)$. Since $D^s\cup Z\subseteq X^R$, by Lemma \ref{Lma_IR} (ii) we have $D^s\cup Z\subseteq D^s$. This contradicts that $Z$ is nonempty and $Z\subseteq X\setminus(A^s\cup D^s)$.
\end{proof}

\begin{proof}[Proof of Proposition \ref{prop_equi}]
We first show that, when each central agent has to pick one side, an outcome is blocked only if it is blocked by a set of contracts on one side, and an outcome is setwise blocked only if it is setwise blocked by a set of contracts on one side.
\begin{lemma}\label{lma_block}
\normalfont
Suppose each central agent $i\in M$ has to pick one side.
\begin{description}
  \item[(i)] If an outcome $Y\subseteq X$ is blocked by a nonempty set $Z\subseteq X\setminus Y$, then $Y$ is blocked by $Z^L$ if $Z^L\neq\emptyset$, and blocked by $Z^R$ if  $Z^R\neq\emptyset$.
  \item[(ii)] If an outcome $Y\subseteq X$ is setwise blocked by a nonempty set $Z\subseteq X\setminus Y$, then $Y$ is setwise blocked by $Z^L$ if $Z^L\neq\emptyset$, and setwise blocked by $Z^R$ if  $Z^R\neq\emptyset$.
\end{description}
\end{lemma}

\begin{proof}
Suppose $Z^L\neq\emptyset$.

(i) For any central agent $i\in \mathrm{N}(Z^L)\cap I^M$, since $\emptyset\neq Z_i^L\subseteq Z_i\subseteq \mathrm{Ch}_i(Y\cup Z)$ and the agent $i$ has to pick one side, we know that there are no right-side contracts of $i$ in $Z$, and thus, we have $Z^L_i\subseteq \mathrm{Ch}_i(Y\cup Z^L)$. For each left-side agent $i\in \mathrm{N}(Z^L)\cap I^L$, since $Z_i=Z_i^L$, we also have $Z^L_i\subseteq \mathrm{Ch}_i(Y\cup Z^L)$. Therefore, the outcome $Y$ is blocked by $Z^L$. The other part follows the same argument.

(ii) Let $Y^*\subseteq Y\cup Z$ be the outcome such that $Z\subseteq Y^*$, $Y_i^*\succ_iY_i$, and $Y^*_i$ is individually rational for all $i\in \mathrm{N}(Z)$.
For any central agent $i\in \mathrm{N}(Z^L)\cap I^M$, since the agent $i$ has to pick one side and $\emptyset\neq Z^L_i\subseteq Z_i\subseteq Y^*_i=\mathrm{Ch}_i(Y^*_i)$, we know that there are no right-side contracts of $i$ in $Y^*$, and thus, we have $Y_i^{*L}\succ_iY_i$ and $Y^{*L}_i=\mathrm{Ch}_i(Y^{*L}_i)$. For any left-side agent $i\in \mathrm{N}(Z^L)\cap I^L$, we also have $Y_i^{*L}\succ_iY_i$ and $Y^{*L}_i=\mathrm{Ch}_i(Y^{*L}_i)$; then since $Z^L\subseteq Y^{*L}$, the outcome $Y$ is setwise blocked by $Z^L$. The other part follows the same argument.
\end{proof}

The ``only if'' part. Suppose an outcome $Y$ is stable but setwise blocked by $Z$. We have $Z^L\neq\emptyset$ or $Z^R\neq\emptyset$. Without loss of generality, we assume $Z^L\neq\emptyset$, then by Lemma \ref{lma_block}, $Y$ is also setwise blocked by $Z^L$. Let $Y^*\subseteq Y\cup Z^L$ be the outcome such that $Z^L\subseteq Y^*$, $Y_i^*\succ_iY_i$, and $Y^*_i$ is individually rational for all $i\in \mathrm{N}(Z^L)$. For all $i\in \mathrm{N}(Z^L)$, since $\emptyset\neq Z^L_i\subseteq Y^*_i=\mathrm{Ch}_i(Y^*_i)$, and each central agent has to pick one side, we have $Y_i^*\subseteq X^L$; and thus, $Y_i^*\succ_iY_i=\mathrm{Ch}_i(Y_i)$ and $Y_i^*\subseteq Y_i\cup Z_i^L$ imply $\mathrm{Ch}_i(Y_i\cup Z_i^L)\subseteq X^L$. We have $Z_i^L\subseteq Y^*_i= \mathrm{Ch}_i(Y^*_i)\subseteq Y^L_i\cup Z_i^L$; thus, same-side complementarity implies $Z_i^L\subseteq \mathrm{Ch}_i(Y_i^L\cup Z_i^L)=\mathrm{Ch}_i(Y_i\cup Z_i^L)$. Hence, we know that $Z^L$ blocks $Y$, which contradicts the stability of $Y$. 

The ``if'' part. Suppose an outcome $Y$ is setwise stable but blocked by $Z$. We have $Z^L\neq\emptyset$ or $Z^R\neq\emptyset$. Without loss of generality, we assume $Z^L\neq\emptyset$, then by Lemma \ref{lma_block}, $Y$ is also blocked by $Z^L$. For any $i\in \mathrm{N}(Z^L)$, since $\emptyset\neq Z_i^L\subseteq\mathrm{Ch}_i(Y_i\cup Z_i^L)$, and each central agent has to pick one side, we have $\mathrm{Ch}_i(Y_i\cup Z_i^L)\subseteq X^L$. For any $i\in \mathrm{N}(Z^L)$, since $Y^L_i\subseteq Y_i=\mathrm{Ch}_i(Y_i)$, same-side complementarity implies $Y^L_i\subseteq\mathrm{Ch}_i(Y_i\cup Z_i^L)$; and thus, $\mathrm{Ch}_i(Y_i\cup Z_i^L)\subseteq X^L$ implies $\mathrm{Ch}_i(Y_i\cup Z_i^L)=Y^L_i\cup Z_i^L$. Let $Y^*\equiv \bigcup_{i\in \mathrm{N}(Z^L)}(Y^L_i\cup Z_i^L)$, we have $Y^*\subseteq Y\cup Z^L$ and $Z^L\subseteq Y^*$. For any $i\in \mathrm{N}(Z^L)$, we have $Y^*_i=Y^L_i\cup Z_i^L=\mathrm{Ch}_i(Y^*_i)$ and $Y^*_i=\mathrm{Ch}_i(Y_i\cup Z_i^L)\succ_iY_i$. Therefore, $Y$ is setwise blocked by $Z^L$. 
\end{proof}

\subsection{An example for Section \ref{Sec_pick}}\label{Sec_exam}

In this section, we provide an example of the alternate DA algorithm for the organization-member market described in Section \ref{Sec_pick}. We switch to the notation of \cite{RS90}.

\begin{example}\label{exam_org}
\normalfont
Consider a market with two organizations ($o_1$ and $o_2$) and five agents ($i_1, i_2, i_3, i_4$, and $i_5$).
\begin{align*}
  o_1: \{i_1,i_2,i_3,i_4,i_5\}\succ\{i_1,i_4,i_5\}\succ\{i_2,i_3\}\succ\emptyset\qquad\qquad\qquad & i_1: o_2\succ o_1\succ\emptyset \\
    o_2: \{i_1,i_2,i_3,i_4\}\succ\{i_2,i_3\}\succ\{i_2\}\succ\emptyset\qquad\qquad\qquad & i_2: o_2\succ o_1\succ\emptyset \\
  & i_3: o_1\succ o_2\succ\emptyset \\
  & i_4: o_1\succ o_2\succ\emptyset \\
  & i_5: o_1\succ\emptyset
\end{align*}

\textbf{Stage} 1. All agents propose to the organization $o_1$, and $o_1$ accepts all proposals. Then, the agents $i_1$ and $i_2$ propose to the organization $o_2$, and $o_2$ accepts $i_2$ and rejects $i_1$.

\textbf{Stage} 2. All agents except $i_2$ propose to $o_1$, and $o_1$ accepts $i_1$, $i_4$, and $i_5$ and rejects $i_3$. Then $i_1$, $i_2$, and $i_3$ propose to $o_2$; and $o_2$ accepts $i_2$ and $i_3$ and rejects $i_1$. The algorithm terminates since the chosen set of $o_2$ does not contain any agent currently matching to $o_1$.

The algorithm matches $\{i_1,i_4,i_5\}$ to $o_1$ and matches $\{i_2,i_3\}$ to $o_2$. If we swap the positions of the two organizations in the algorithm, the algorithm produces another outcome: matching $\{i_1,i_2,i_3,i_4\}$ to $o_2$, and leaving $i_5$ and $o_1$ unmatched.
\end{example}

\subsection{Proofs for Section \ref{Sec_TU}}

\begin{proof}[Proof of Lemma \ref{lma_stable}]
Suppose $(\Phi,\mathbf{p})$ is a competitive equilibrium, but $\kappa(\Phi,\mathbf{p})$ is not stable.

(i) Suppose $\kappa(\Phi,\mathbf{p})$ is not individually rational for some $i\in I$: $[\kappa(\Phi,\mathbf{p})]_i\neq \mathrm{Ch}_i(\kappa(\Phi,\mathbf{p}))$. This contradicts $\Phi_i\in\mathrm{D}_i(\mathbf{p}_i)$.

(ii) Suppose $\kappa(\Phi,\mathbf{p})$ is blocked by $Z\subseteq X\setminus\kappa(\Phi,\mathbf{p})$. Then, for each $i\in\mathrm{N}(Z)$, there exists $A^i\in \mathrm{Ch}_i(\kappa(\Phi,\mathbf{p})\cup Z)$ such that $Z_i\subseteq A^i$ and $\mathrm{u}_i(A^i)>\mathrm{u}_i([\kappa(\Phi,\mathbf{p})]_i)$. For each $i\in\mathrm{N}(Z)$, let $\mathrm{t}_i(Z_i)\equiv\sum_{(w,\mathbf{\hat{p}}^w)\in Z_i}\hat{p}^w_i$ be her net payment in the contracts of $Z_i$, and let $\mathrm{t}'_i(Z_i,\mathbf{p})\equiv\sum_{(w,\mathbf{\hat{p}}^w)\in Z_i}p^w_i$ be her net payment from the primitive contracts of $\tau(Z_i)$ if she signs these primitive contracts at $\mathbf{p}$. Suppose $\mathrm{t}_i(Z_i)\geq\mathrm{t}'_i(Z_i,\mathbf{p})$ for some $i\in\mathrm{N}(Z)$. The block means that the agent $i$ can obtain a larger utility $\mathrm{u}_i(A^i)$ by signing the contracts of $Z_i$ while possibly dropping some original contracts. Hence, she can obtain an even larger utility by signing the primitive contracts of $\tau(Z_i)$ at $\mathbf{p}$ while dropping the same contracts. This contradicts that $(\Phi,\mathbf{p})$ is a competitive equilibrium. Thus, we have $\mathrm{t}_i(Z_i)<\mathrm{t}'_i(Z_i,\mathbf{p})$ for all $i\in\mathrm{N}(Z)$. However, this is impossible since $\sum_{i\in\mathrm{N}(Z)}\mathrm{t}_i(Z_i)=\sum_{i\in\mathrm{N}(Z)}\mathrm{t}'_i(Z_i,\mathbf{p})=0$.
\end{proof}

\begin{proof}[Proof of Lemma \ref{lma_RGC}]
(iv)$\Rightarrow$(i). Recall that  $\mathrm{v}_i$ is supermodular if and only if for any $\Phi\subset\Psi\subseteq \Omega_i$ and $w'\in\Omega_i\setminus\Psi$, $\mathrm{v}_i(\Phi\cup\{w'\})-\mathrm{v}_i(\Phi)\leq \mathrm{v}_i(\Psi\cup\{w'\})-\mathrm{v}_i(\Psi)$. Suppose $\mathrm{v}_i$ is not supermodular, then there exist $\Phi\subset\Psi\subseteq \Omega_i$ and $w'\in\Omega_i\setminus\Psi$ such that $\mathrm{v}_i(\Phi\cup\{w'\})-\mathrm{v}_i(\Phi)> \mathrm{v}_i(\Psi\cup\{w'\})-\mathrm{v}_i(\Psi)$. Let $H\equiv\sum_{\Phi'\subseteq\Omega_i}|\mathrm{v}_i(\Phi')|$ be a sufficiently large number. Select $\mathbf{p}_i\in\mathbb{R}^{\Omega_i}$ such that $p_i^w=H$ for each $w\notin\Phi\cup\{w'\}$, $p_i^w=-H$ for each $w\in\Phi$, and $p^{w'}_i$ is a number satisfying
\begin{equation}\label{GS}
\mathrm{v}_i(\Phi\cup\{w'\})-\mathrm{v}_i(\Phi)>p^{w'}_i> \mathrm{v}_i(\Psi\cup\{w'\})-\mathrm{v}_i(\Psi),
\end{equation}
The first inequality of (\ref{GS}) implies
\begin{equation}\label{GS1}
\mathrm{v}_i(\Phi\cup\{w'\})-\sum_{w\in\Phi\cup\{w'\}}p^w_i>\mathrm{v}_i(\Phi)-\sum_{w\in\Phi}p^w_i. 
\end{equation}
Since $H$ is sufficiently large, $p_i^w=H$ for each $w\notin\Phi\cup\{w'\}$, and $p_i^w=-H$ for each $w\in\Phi$, we know that (\ref{GS1}) implies $\mathrm{D}_i(\mathbf{p}_i)=\{\Phi\cup\{w'\}\}$.

Let $\mathbf{q}_i\in\mathbb{R}^{\Omega_i}$ be the price vector with the same coordinates as $\mathbf{p}_i$ except that $q_i^w=-H$ for each $w\in\Psi\setminus\Phi$. The second inequality of (\ref{GS}) implies
\begin{equation}\label{GS2} \mathrm{v}_i(\Psi)-\sum_{w\in\Psi}q^w_i>\mathrm{v}_i(\Psi\cup\{w'\})-\sum_{w\in\Psi\cup\{w'\}}q^w_i,
\end{equation}
Since $H$ is sufficiently large, $q_i^w=H$ for each $w\notin\Psi\cup\{w'\}$, and $q_i^w=-H$ for each $w\in\Psi$, we know that (\ref{GS2}) implies $\mathrm{D}_i(\mathbf{q}_i)=\{\Psi\}$. 

Consequently, condition (iv) fails: We have $\mathbf{p}_i\geq\mathbf{q}_i$ and $|\mathrm{D}_i(\mathbf{p}_i)|=|\mathrm{D}_i(\mathbf{q}_i)|=1$, but we also have $p_i^{w'}=q_i^{w'}$, $w'\in\Phi\cup\{w'\}$, and $w'\notin\Psi$. 
\end{proof}

\begin{proof}[Proof of Lemma \ref{lma_condition}]
Let $\tilde{\mathrm{D}}_i$ be the demand correspondence of each central or right-side agent $i\in I^M\cup I^R$ in the modified market. We abuse the notation by using $\mathrm{g}(\cdot)$ on a price vector in $\mathbb{R}^{\Omega_i}$ to reverse the directions of the transfers of all right-side contracts: For every central or right-side agent $i\in I^M\cup I^R$ and $\mathbf{p}_i\in \mathbb{R}^{\Omega_i}$, let $\mathrm{g}(\mathbf{p}_i)\in\mathbb{R}^{\Omega_i}$ be the price vector whose $w$-wise component is $p^w_i$ if $w\in \Omega_i^L$ and $-p^w_i$ if $w\in \Omega_i^R$. Notice that $\mathrm{g}(\mathrm{g}(\mathbf{p}))=\mathbf{p}$ for all $\mathbf{p}\in\mathbb{B}$, and $\mathrm{g}(\mathrm{g}(\mathbf{p}_i))=\mathbf{p}_i$ for all $\mathbf{p}_i\in \mathbb{R}^{\Omega_i}$. The following lemma shows the relation between the demand correspondences of the original market and the modified market.

\begin{lemma}\label{lma_relation}
\normalfont
\begin{description}
  \item[(i)] For each central agent $i\in I^M$ and  $\mathbf{p}_i\in \mathbb{R}^{\Omega_i}$, we have $\Psi\in \tilde{\mathrm{D}}_i(\mathrm{g}(\mathbf{p}_i))$ if and only if $\Psi^L\cup(\Omega_i^R\setminus \Psi^R)\in\mathrm{D}_i(\mathbf{p}_i)$.
  \item[(ii)] For each right-side agent $i\in I^R$ and  $\mathbf{p}_i\in \mathbb{R}^{\Omega_i}$, we have $\Psi\in \tilde{\mathrm{D}}_i(\mathrm{g}(\mathbf{p}_i))$ if and only if $\Omega_i\setminus \Psi\in\mathrm{D}_i(\mathbf{p}_i)$.
\end{description}
\end{lemma}

\begin{proof}
(i) For each central agent $i\in I^M$ and $\mathbf{p}_i\in \mathbb{R}^{\Omega_i}$, notice that $\Psi\in \tilde{\mathrm{D}}_i(\mathrm{g}(\mathbf{p}_i))$ is equivalent to
\begin{equation*}
\mathrm{v}_i(\Psi^L\cup(\Omega_i^R\setminus \Psi^R))-\sum_{w\in\Psi^L}p^w_i+\sum_{w\in\Psi^R}p^w_i\geq \mathrm{v}_i(\Phi^L\cup(\Omega_i^R\setminus \Phi^R))-\sum_{w\in\Phi^L}p^w_i+\sum_{w\in\Phi^R}p^w_i
\end{equation*}
for all $\Phi\in\Omega_i$, which is equivalent to
\begin{equation*}
\mathrm{v}_i(\Psi^L\cup(\Omega_i^R\setminus \Psi^R))-\sum_{w\in\Psi^L\cup(\Omega_i^R\setminus \Psi^R)}p^w_i\geq \mathrm{v}_i(\Phi^L\cup(\Omega_i^R\setminus \Phi^R))-\sum_{w\in\Phi^L\cup(\Omega_i^R\setminus \Phi^R)}p^w_i
\end{equation*}
for all $\Phi\in\Omega_i$. The latter is exactly $\Psi^L\cup(\Omega_i^R\setminus \Psi^R)\in\mathrm{D}_i(\mathbf{p}_i)$.

(ii) For each right-side agent $i\in I^R$ and  $\mathbf{p}_i\in \mathbb{R}^{\Omega_i}$, notice that $\Psi\in \tilde{\mathrm{D}}_i(\mathrm{g}(\mathbf{p}_i))$ is equivalent to
\begin{equation*}
\mathrm{v}_i(\Omega_i\setminus\Psi)+\sum_{w\in\Psi}p^w_i\geq \mathrm{v}_i(\Omega_i\setminus\Phi)+\sum_{w\in\Phi}p^w_i
\end{equation*}
for all $\Phi\in\Omega_i$, which is equivalent to
\begin{equation*}
\mathrm{v}_i(\Omega_i\setminus\Psi)-\sum_{w\in\Omega_i\setminus\Psi}p^w_i\geq \mathrm{v}_i(\Omega_i\setminus\Phi)-\sum_{w\in\Omega_i\setminus\Phi}p^w_i
\end{equation*}
for all $\Phi\in\Omega_i$. The latter is exactly $\Omega_i\setminus \Psi\in\mathrm{D}_i(\mathbf{p}_i)$.
\end{proof}

We are now ready to prove Lemma \ref{lma_condition}.

(i) Fix a right-side agent $i\in I^R$ and two prices $\mathbf{p}_i,\mathbf{q}_i\in\mathbb{R}^{\Omega_i}$ with $\mathbf{p}_i\geq\mathbf{q}_i$. For any $\Psi\in \tilde{\mathrm{D}}_i(\mathbf{q}_i)$, by Lemma \ref{lma_relation} we have $\Omega_i\setminus \Psi\in\mathrm{D}_i(\mathrm{g}(\mathbf{q}_i))$.

Since $\mathrm{g}(\mathbf{q}_i)=-\mathbf{q}_i\geq-\mathbf{p}_i=\mathrm{g}(\mathbf{p}_i)$, gross complementarity (of the form (iii) in Lemma \ref{lma_RGC}) implies that there exists $\Phi\in\mathrm{D}_i(\mathrm{g}(\mathbf{p}_i))$ such that $\{w\in\Omega_i\setminus \Psi|p^w_i=q^w_i\}\subseteq\Phi$.

By Lemma \ref{lma_relation}, we have $\Omega_i\setminus \Phi\in\tilde{\mathrm{D}}_i(\mathbf{p}_i)$. For any $w\in\Omega_i\setminus \Phi$ satisfying $p^w_i=q^w_i$, we have $w\in\Psi$ since $w\in\Omega_i\setminus \Psi$ implies $w\in\Phi$ (according to $\{w\in\Omega_i\setminus \Psi|p^w_i=q^w_i\}\subseteq\Phi$).
Hence, we are done according to (iv)$\Leftrightarrow$(v) of Lemma \ref{lma_RGC}.

(ii) Fix a central agent $i\in I^M$ and two price vectors $\mathbf{p}_i,\mathbf{q}_i\in\mathbb{R}^{\Omega_i}$ such that $\mathbf{p}_i\geq\mathbf{q}_i$ and $|\tilde{\mathrm{D}}_i(\mathbf{p}_i)|=|\tilde{\mathrm{D}}_i(\mathbf{p}_i)|=1$.
Let $\tilde{\mathrm{D}}_i(\mathbf{p}_i)=\{\Phi\}$ and $\tilde{\mathrm{D}}_i(\mathbf{q}_i)=\{\Psi\}$. Let $\hat{\mathbf{p}}_i\in\mathbb{R}^{\Omega_i}$ be the price vector such that $\hat{p}^w_i=p^w_i$ for all $w\in\Omega^L_i$ and $\hat{p}^w_i=q^w_i$ for all $w\in\Omega^R_i$. Select $\hat{\Phi}\in\mathrm{D}_i(\mathrm{g}(\hat{\mathbf{p}}_i))$ in the original market. Let $\epsilon>0$ be a sufficiently small number, and let $\hat{\mathbf{q}}_i,\tilde{\mathbf{p}}_i,\tilde{\mathbf{q}}_i\in\mathbb{R}^{\Omega_i}$ be the price vectors such that
\begin{align}
\hat{q}^w_i=\left\{
\begin{aligned}\label{qhat}
&[\mathrm{g}(\hat{\mathbf{p}}_i)]^w-\epsilon \qquad &\text{ if } w\in\hat{\Phi}, \\
&[\mathrm{g}(\hat{\mathbf{p}}_i)]^w+\epsilon \qquad &\text{ if } w\in\Omega_i\setminus\hat{\Phi},
\end{aligned}
\right.\\
\tilde{p}^w_i=\left\{
\begin{aligned}\label{pperb}
&p_i^w-\epsilon \qquad &\text{ if } w\in\hat{\Phi}^L, \\
&p_i^w+\epsilon \qquad &\text{ if } w\in\Omega^L_i\setminus\hat{\Phi}^L, \\
&p_i^w+\epsilon \qquad &\text{ if } w\in\hat{\Phi}^R, \\
&p_i^w-\epsilon \qquad &\text{ if } w\in\Omega^R_i\setminus\hat{\Phi}^R, 
\end{aligned}
\right.\\
\tilde{q}^w_i=\left\{
\begin{aligned}\label{qperb}
&q_i^w-\epsilon \qquad &\text{ if } w\in\hat{\Phi}^L, \\
&q_i^w+\epsilon \qquad &\text{ if } w\in\Omega^L_i\setminus\hat{\Phi}^L, \\
&q_i^w+\epsilon \qquad &\text{ if } w\in\hat{\Phi}^R, \\
&q_i^w-\epsilon \qquad &\text{ if } w\in\Omega^R_i\setminus\hat{\Phi}^R, 
\end{aligned}
\right.
\end{align}
Notice that $\hat{\Phi}\in\mathrm{D}_i(\mathrm{g}(\hat{\mathbf{p}}_i))$ and (\ref{qhat}) imply $\mathrm{D}_i(\hat{\mathbf{q}}_i)=\{\hat{\Phi}\}$; when $\epsilon$ is sufficiently small, (\ref{pperb}) further implies

(a) $\tilde{\mathrm{D}}_i(\tilde{\mathbf{p}}_i)=\tilde{\mathrm{D}}_i(\mathbf{p}_i)=\{\Phi\}$, which by Lemma \ref{lma_relation} implies $\mathrm{D}_i(\mathrm{g}(\tilde{\mathbf{p}}_i))=\{\Phi^L\cup(\Omega_i^R\setminus \Phi^R)\}$, and

(b) $[\mathrm{g}(\tilde{\mathbf{p}}_i)]^w=\hat{q}_i^w$ for all $w\in\Omega_i^L$, and $[\mathrm{g}(\tilde{\mathbf{p}}_i)]^w\leq\hat{q}_i^w$ for all $w\in\Omega_i^R$, where the equality holds when $p_i^w=q_i^w$,\\
and (\ref{qperb}) further implies

(c) $\tilde{\mathrm{D}}_i(\tilde{\mathbf{q}}_i)=\tilde{\mathrm{D}}_i(\mathbf{q}_i)=\{\Psi\}$, which by Lemma \ref{lma_relation} implies $\mathrm{D}_i(\mathrm{g}(\tilde{\mathbf{q}}_i))=\{\Psi^L\cup(\Omega_i^R\setminus \Psi^R)\}$, and

(d) $[\mathrm{g}(\tilde{\mathbf{q}}_i)]^w=\hat{q}_i^w$ for all $w\in\Omega_i^R$ and $[\mathrm{g}(\tilde{\mathbf{q}}_i)]^w\leq\hat{q}_i^w$ for all $w\in\Omega_i^L$, where the equality holds when $p_i^w=q_i^w$.

According to (a), (b), and $\mathrm{D}_i(\hat{\mathbf{q}}_i)=\{\hat{\Phi}\}$, same-side gross complementarity and cross-side gross substitutability imply $\{w\in\hat{\Phi}^R|p_i^w=q_i^w\}\subseteq\Omega_i^R\setminus \Phi^R$ and $\Phi^L\subseteq\hat{\Phi}^L$.

According to (c), (d), and $\mathrm{D}_i(\hat{\mathbf{q}}_i)=\{\hat{\Phi}\}$, same-side gross complementarity and cross-side gross substitutability imply $\{w\in\hat{\Phi}^L|p_i^w=q_i^w\}\subseteq\Psi^L$ and $\Omega_i^R\setminus\Psi^R\subseteq\hat{\Phi}^R$.

Primitive contracts are gross complements for agent $i$ in the modified market since 

$\bullet$ if $w\in\Phi^L$ satisfies $p^w_i=q^w_i$, then $\Phi^L\subseteq\hat{\Phi}^L$ and $\{w\in\hat{\Phi}^L|p_i^w=q_i^w\}\subseteq\Psi^L$ imply $w\in\Psi^L$, and

$\bullet$ if $w\in\Phi^R$ satisfies $p^w_i=q^w_i$, then $\Omega_i^R\setminus\Psi^R\subseteq\hat{\Phi}^R$ and $\{w\in\hat{\Phi}^R|p_i^w=q_i^w\}\subseteq\Omega_i^R\setminus \Phi^R$ imply $w\in\Psi^R$.
\end{proof}

\begin{proof}[Proof of Lemma \ref{lma_CE}]
Since $(\Psi,\mathbf{p})$ is a competitive equilibrium in the modified market, and each left-side agent $i\in I^L$ has the same demand correspondence in the two markets, we have $\Psi^L_i=\Psi_i\in\mathrm{D}_i(\mathbf{p}_i)=\mathrm{D}_i([\mathrm{g}(\mathbf{p})]_i)$. For each central agent $i\in I^M$, Lemma \ref{lma_relation} and $\Psi_i\in\tilde{\mathrm{D}}_i(\mathbf{p}_i)$ imply $\Psi_i^L\cup(\Omega_i^R\setminus \Psi_i^R)\in\mathrm{D}_i(\mathrm{g}(\mathbf{p}_i))$; and for each right-side agent $i\in I^R$, Lemma \ref{lma_relation} and $\Psi_i\in\tilde{\mathrm{D}}_i(\mathbf{p}_i)$ imply $\Omega_i\setminus \Psi_i\in\mathrm{D}_i(\mathrm{g}(\mathbf{p}_i))$. Hence, $(\Psi^L\cup(\Omega^R\setminus\Psi^R),\mathrm{g}(\mathbf{p}))$ is a competitive equilibrium in the original market.
\end{proof}


\bigskip


\begin{thebibliography}{99999999999999999999999999999999999999999}

\bibitem[Azevedo and Hatfield(2018)]{AH18}{\small Azevedo, E. M., Hatfield, J. M., 2018. Existence of equilibrium in large matching markets with complementarities. \emph{Working paper}, https://ssrn.com/abstract=3268884.}
    
\bibitem[Azevedo et al.(2013)]{AWW13}{\small Azevedo, E. M., Weyl, E. G., White, A., 2013. Walrasian equilibrium in large, quasilinear markets. \emph{Theoretical Economics}, 8(2), 281-290. }
    
\bibitem[Baldwin et al.(2023)]{BEJK23}{\small Baldwin, E., Edhan, O., Jagadeesan, R., Klemperer, P., Teytelboym, A., 2023. The equilibrium existence duality. \emph{Journal of Political Economy}, 131(6), 1440-1476.}

\bibitem[Baldwin and Klemperer(2019)]{BK19}{\small Baldwin, E., Klemperer, P., 2019. Understanding Preferences: ``Demand Types'', and the Existence of Equilibrium with Indivisibilities. \emph{Econometrica}, 87, 867-932. }

\bibitem[Bando and Hirai(2021)]{BH21}{\small Bando, K., Hirai, T., 2021. Stability and venture structures in multilateral matching. \emph{Journal of Economic Theory}, 196, 105292.}
    
\bibitem[Che et al.(2019)]{CKK19}{\small Che, Y-K., Kim, J., Kojima, F., 2019. Stable matching in large economies. \textit{Econometrica}, 87(1), 65-110.}

\bibitem[Danilov et al.(2001)]{DKM01}{\small Danilov, V., Koshevoy, G., Murota, K., 2001. Discrete convexity and equilibria in economics with indivisible goods and money. \emph{Mathematical Social Sciences}, 41, 251-273.}

\bibitem[Echenique and Oviedo(2006)]{EO06}{\small Echenique, F., Oviedo, J., 2006. A theory of stability in many-to-many matching. \textit{Theoretical Economics}, 1, 233-273.}
    
\bibitem[Echenique and Yenmez(2007)]{EY07}{\small Echenique, F., Yenmez, M. B., 2007. A solution to matching with preferences over colleagues. \textit{Games and Economic Behavior}, 59, 46-71.}

\bibitem[Fleiner et al.(2019)]{FJJT19}{\small Fleiner, T., Jagadeesan, R., Jank\'{o}, Z., Teytelboym, A., 2019. Trading networks with frictions. \textit{Econometrica}, 87, 1633-1661.}
  
\bibitem[Gul and Stacchetti(1999)]{GS99}{\small Gul, F., Stacchetti, E., 1999. Walrasian equilibrium with gross substitutes. \textit{Journal of Economic Theory}, 87, 95-124.}
    
\bibitem[Hatfield and Kojima(2008)]{HK08}{\small Hatfield, J. W., Kojima, F., 2008. Matching with contracts: Comment. \textit{American Economic Review}, 98(3), 1189-1194.}
    
\bibitem[Hatfield and Kojima(2010)]{HK10}{\small Hatfield, J. W., Kojima, F., 2010. Substitutes and stability for matching with contracts. \textit{Journal of Economic Theory}, 145(5), 1704-1723.}
    
\bibitem[Hatfield and Kominers(2012)]{HK12}{\small Hatfield, J. W., Kominers, S.D., 2012. Matching in networks with bilateral contracts. \textit{American Economic Journal: Miroeconomics}, 4, 176-208.}

\bibitem[Hatfield et al.(2013)]{HKNOW13}{\small Hatfield, J. W., Kominers, S. D.,  Nichifor, A., Ostrovsky, M., Westkamp, A., 2013. Stability and competitive equilibrium in trading networks. \textit{Journal of Political Economy}, 121(5), 966-1005.}
    
\bibitem[Hatfield et al.(2019)]{HKNOW19}{\small Hatfield, J. W., Kominers, S. D.,  Nichifor, A., Ostrovsky, M., Westkamp, A., 2019. Full substitutability. \textit{Theoretical Economics}, 14, 1535-1590.}

\bibitem[Hatfield and Milgrom(2005)]{HM05}{\small Hatfield, J. W., Milgrom, P. R., 2005. Matching with contracts. \textit{American Economic Review}, 95, 913-935.}
        
\bibitem[Huang(2023a)]{H23a}{\small Huang, C., 2023a. Stable matching: An integer programming approach. \emph{Theoretical Economics}, 18, 37-63.}
    
\bibitem[Huang(2023b)]{H23b}{\small Huang, C., 2023b. Multilateral matching with scale economies. \emph{Working paper}, arXiv:2310.19479.}
    
\bibitem[Kelso and Crawford(1982)]{KC82}{\small Kelso, A. S., Crawford, V. P., 1982. Job matching, coalition formation and gross substitutes. \textit{Econometrica}, 50, 1483-1504. }
    
\bibitem[Klaus and Klijn(2005)]{KK05}{\small Klaus, B., Klijn, F., 2005. Stable Matchings and Preferences of Couples. \textit{Journal of Economic Theory}, 121, 75-106.}

\bibitem[Klaus and Walzl(2009)]{KW09}{\small Klaus, B., Walzl, M., 2009. Stable many-to-many matchings with contracts. \textit{Journal of Mathematic Economics}, 45, 422-434.}
    
\bibitem[Kojima et al.(2013)]{KPR13}{\small Kojima, F., Pathak, P.A., Roth, A.E., 2013. Matching with Couples: Stability and Incentives in Large Markets. \textit{Quarterly Journal of Economics}, 128, 1585-1632.}
    
\bibitem[Nguyen and Vohra(2018)]{NV18}{\small Nguyen, T., Vohra, R., 2018. Near-Feasible Stable Matchings with Couples. \textit{American Economic Review}, 108(11), 3154-3169.}
    
\bibitem[Nguyen and Vohra(2019)]{NV19}{\small Nguyen, T., Vohra, R., 2019. Stable matching with proportionality constraints. \textit{Operations Research}, 67(6), 1503-1519.}
    
\bibitem[Nguyen and Vohra(2024)]{NV24}{\small Nguyen, T., Vohra, R., 2024. (Near) Substitute preferences and equilibria with indivisibilities. \textit{Journal of Political Economy}, 132(12), 4122-4154.}

\bibitem[Ostrovsky(2008)]{O08}{\small Ostrovsky, M., 2008. Stability in supply chain networks. \textit{American Economic Review}, 98, 897-923. }
    
\bibitem[Pycia(2012)]{P12}{\small Pycia, M., 2012. Stability and preference alignment in matching and coalition formation. \textit{Econometrica}, 80(1), 323-362. }
    
\bibitem[Pycia and Yenmez(2023)]{PY23}{\small Pycia, M., Yenmez, M.B., 2023. Matching with externalities. \textit{Review of Economic Studies}, 90, 948-974.}

\bibitem[Rostek and Yoder(2020)]{RY20}{\small Rostek, M., Yoder, N., 2020. Matching with complementary contracts. \textit{Econometrica}, 88(5), 1793-1824.}

\bibitem[Rostek and Yoder(2025)]{RY25}{\small Rostek, M., Yoder, N., 2025. Complementarity in matching markets and exchange economies. \textit{Games and Economic Behavior}, 150, 415-435.}
    
\bibitem[Roth(1984)]{R84}{\small Roth, A. E., 1984. Stability and Polarization of Interests in Job Matching. \emph{Econometrica}, 52, 47-57.}

\bibitem[Roth and Sotomayor(1990)]{RS90}{\small Roth, A. E., Sotomayor, M., 1990. Two-sided matching: A study in game-theoretic modelling and analysis. Econometric Society Monographs No. 18, Cambridge University Press, Cambridge England.}
    
\bibitem[S\"{o}nmez(2013)]{S13}{\small S\"{o}nmez, T., 2013. Bidding for army career specialties: Improving the ROTC branching mechanism. \textit{Journal of Political Economy}, 121, 186-219.}
    
\bibitem[S\"{o}nmez and Switzer(2013)]{SS13}{\small S\"{o}nmez, T., Switzer, T. B., 2013. Matching with (branch-of-choice) contracts at United States Military Academy. \textit{Econometrica}, 81, 451-488.}
    
\bibitem[Sotomayor(1999)]{S99}{\small Sotomayor, M. A. O., 1999. Three remarks on the many-to-many stable matching problem. \textit{Mathematical Social Sciences}, 38, 55-70.}
    
\bibitem[Sun and Yang(2006)]{SY06}{\small Sun, N., Yang, Z., 2006. Equilibria and indivisibilities: Gross substitutes and complements. \textit{Econometrica}, 74, 1385-1402.}

\bibitem[Yokote(2023)]{Y23}{\small Yokote, K., 2023. A critical comparison between the gross substitutes and complements conditions. \textit{Economics Letters}, 226, 111106.}
\end{thebibliography}
\end{document}